\documentclass[a4paper,USenglish,cleveref,thm-restate]{lipics-v2021}
\usepackage{tikz}
\usepackage{enumitem}

\definecolor{varyellow}{RGB}{253,216,53}

\newcommand{\avalzero}{\mathsf{A}}
\newcommand{\avalone}{\mathsf{B}}
\newcommand{\avaltwo}{\mathsf{C}}
\newcommand{\avalthree}{\mathsf{D}}
\newcommand{\avalfour}{\mathsf{E}}

\newcommand{\flip}[3]{#1\colon #2\rightsquigarrow #3}

\EventShortTitle{STACS 2027}
\EventAcronym{STACS}
\EventYear{2027}

\nolinenumbers

\makeatletter
\AtBeginDocument{%
  \ifx\authoranonymous\relax
    \gdef\@author{%
      {\Large\bfseries\textcolor{red}{Anonymous author(s)}\par}%
      {\small\textcolor{red}{Anonymous affiliation(s)}\par}%
    }%
  \fi
}
\makeatother

\title{Exponential Local Search in Bounded-Occurrence MAX-CSP}

\author{Sophia Heimann}%
{Hertz Chair for Algorithms and Optimization, University of Bonn, Germany}%
{sheimann@uni-bonn.de}%
{}
{}
\author{Hung P. Hoang}%
{Algorithms and Complexity Group, Faculty of Informatics, TU Wien, Austria}%
{phoang@ac.tuwien.ac.at}%
{https://orcid.org/0000-0001-7883-4134}
{funded by the Austrian Science Fund (FWF, project 10.55776/ESP1136425)}
\author{Stefan Hougardy}%
{Research Institute for Discrete Mathematics, University of Bonn, Germany \and Hausdorff Center for Mathematics, University of Bonn, Germany}%
{hougardy@or.uni-bonn.de}%
{https://orcid.org/0000-0001-8656-3418}%
{funded by the Deutsche Forschungsgemeinschaft (DFG, German Research Foundation) under Germany's Excellence Strategy -- EXC-2047/1 -- 390685813}

\authorrunning{S. Heimann and H. P. Hoang and S. Hougardy}

\usepackage{xcolor}

\begin{document}

\titlerunning{Exponential Local Search in Bounded-Occurrence MAX-CSP}

\Copyright{Sophia Heimann and Hung P. Hoang and Stefan Hougardy}

\hideLIPIcs

\maketitle

\begin{abstract}
We study the number of improving steps that local search may require for $(2,2,r)$-MCA (Maximum Constraint Assignment), a bounded-occurrence version of MAX-CSP with valued constraints in which every constraint involves at most two variables, every variable occurs in at most two constraints, and every variable takes one of $r$ values. We consider the flip neighborhood, where two assignments are neighbors if they differ in the value of exactly one variable; moving to such a neighbor is called a flip. Local search repeatedly performs improving flips, that is, flips that strictly increase the objective value, until a local optimum is reached.

Our main result is that $(2,2,5)$-MCA has the all-exp property: there is an infinite family of instances with
initial assignments for which \emph{every} sequence of improving flips ending in a local optimum has length exponential in the number of variables. Thus, on these instances, local search requires exponentially many steps
regardless of which improving flip is chosen at each step. This is the first all-exp result for $(2,2,r)$-MCA for any value of $r$, and our result extends to every $r\ge 5$. 

In contrast, for $(2,2,2)$-MCA every sequence of improving flips has length at most quadratic in the number of variables~\cite{Kaz2020,KCJ2020}. We give a short, direct proof of this result.  
\end{abstract}

\newpage

\section{Introduction}

Improving a solution step by step through small changes is perhaps the most natural way to approach a hard optimization problem. Local search algorithms are built on exactly this principle, and they are among the most widely used methods in combinatorial optimization. The same principle also appears in processes that were never designed as algorithms. For example, the convergence of Hopfield networks to stable states~\cite{SY1991} and the adaptation of populations on fitness landscapes in evolutionary biology~\cite{Kaz2020,KCJ2020} can both be viewed as local search, in which a system repeatedly makes a small change that improves an objective. Understanding how local search behaves is therefore relevant well beyond algorithm design.

A local search algorithm moves in each iteration from the current solution to an improving solution in its neighborhood until it reaches a local optimum; that is, a solution without an improving neighbor.
A basic question about local search algorithms is how many iterations are needed to reach a local optimum. 
Following~\cite{Monien_Dumrauf_Tscheuschner_2010}, we say that improving sequences are \emph{exponentially long} if their length is at least $2^{\sqrt[k]{|I|}}$ for some fixed $k\in\mathbb{N}$ that is independent of the instance $I$.
A local search problem has the \emph{is-exp property} if there are infinitely many instances, each with an initial solution, from which some sequence of improving steps ending in a local optimum is exponentially long.  
It has the \emph{all-exp property} if there are infinitely many such instances and initial solutions from which \emph{every} sequence of improving steps ending in a local optimum is exponentially long. The all-exp property thus means that local search needs exponentially many iterations on these instances, no matter which improving neighbor is chosen in each iteration. 

In this paper, we ask how simple an optimization problem can be while local search still has the all-exp property. We study this question for maximum constraint satisfaction (MAX-CSP), one of the most general frameworks for combinatorial optimization, with its natural neighborhood. More precisely, we consider $(2,2,r)$-MCA (Maximum Constraint Assignment), the version of MAX-CSP with valued constraints in which every constraint involves at most two variables, every variable occurs in at most two constraints, and every variable takes one of $r$ values. This is about as restricted as a constraint problem can be while still allowing variables to interact. We use the \emph{flip neighborhood}, in which two assignments are neighbors if they differ in the value of exactly one variable; thus, a local search step changes the value of a single variable while leaving all others unchanged.

For binary variables, that is, for domain size $r=2$, local search reaches a local optimum 
after a number of iterations that is at most quadratic in the number of variables, regardless of the weights, as follows from the more general results of~\cite{Kaz2020,KCJ2020}. It is natural to ask whether this remains true for larger domains. Our main result shows that it does not: already with five values per variable, there are
instances and initial assignments from which every sequence of improving steps leading to a local optimum has exponential length. Previously, no all-exp result was known for $(2,2,r)$-MCA with the flip neighborhood for any value of~$r$.

\begin{restatable}{theorem}{Theoremallexp}
\label{thm:225}
    $(2,2,r)$-MCA with the flip neighborhood has the all-exp property  for $r \ge 5$.
\end{restatable}

A standard way to establish the all-exp property for a local search problem is to prove its PLS-completeness via a tight PLS-reduction. The complexity class PLS was introduced by Johnson, Papadimitriou, and Yannakakis to formalize the complexity of finding locally optimal solutions~\cite{JPY1988}. Schäffer and Yannakakis subsequently introduced tight PLS-reductions and showed that such reductions preserve the all-exp property; hence a problem that is PLS-complete under tight reductions has the all-exp property~\cite{SY1991}. 
In our setting, this route to proving the all-exp property is not available, as Dumrauf and Monien~\cite{DM2013} observed that for $(2,2,r)$-MCA even a globally optimal solution can be found in polynomial time; indeed, the constraint graph of such an instance has maximum degree two, so it is a disjoint union of paths and cycles, and an optimal assignment can be computed in polynomial time by dynamic programming. 
Since a global optimum is in particular a local optimum, a local optimum of $(2,2,r)$-MCA can be found in polynomial time. If $(2,2,r)$-MCA were PLS-complete, this would yield polynomial-time algorithms for all problems in PLS. Hence, unless all problems in PLS can be solved in polynomial time, $(2,2,r)$-MCA is not PLS-complete, and~\cref{thm:225} must be proved by a different approach. 

We therefore prove~\cref{thm:225} directly from the definition of the all-exp property. For every $n\geq1$ we construct an instance of $(2,2,5)$-MCA with $2n+1$ variables and $2n+1$ constraints, together with an initial assignment, such that every maximal sequence of improving flips has length exactly $2^{n+3}-7$; the case $r>5$ then follows by adding values that are never used. Since $(2,2,r)$-MCA can be solved in polynomial time, \cref{thm:225} also shows that the all-exp property is a statement about local search itself: these instances are easy for other algorithms, yet local search with the flip neighborhood, started at the given initial assignment, needs exponentially many steps, regardless of which improving flip it chooses in each step.

\subsection{Previous work}

MCA provides a natural extension of weighted satisfiability problems and was studied systematically by Dumrauf and Monien~\cite{DM2013}. Building on techniques introduced by Krentel~\cite{Kre1989}, they proved the PLS-completeness of $(3,2,3)$-MCA and $(2,3,6)$-MCA via a tight reduction, and hence the all-exp property for both problems. They also obtained PLS-completeness via a tight reduction for the binary problem $(6,2,2)$-MCA. On the positive side, Dumrauf and Monien~\cite{DM2013} proved that $(2,2,r)$-MCA can be solved in polynomial time for all values of $r$. Moreover, Kaznatcheev, Cohen, and Jeavons~\cite{KCJ2020} proved that for binary variables, constraints of arity at most two, and a constraint graph that is a tree, every sequence of improving flips has length at most quadratic in the number of variables, independently of the weights; for constraint graphs of maximum degree two, which also include cycles, they sketch the same bound and refer to Kaznatcheev~\cite{Kaz2020} for the full analysis. In particular, every sequence of improving flips for $(2,2,2)$-MCA has at most quadratic length. For $(2,2,3)$-MCA, 
Kaznatcheev, Cohen, and Jeavons~\cite{KCJ2020} proved the is-exp property by an explicit family of instances whose constraint graphs are paths. 

Recently, Kaznatcheev and Volgering~\cite{KV2026} constructed instances with binary variables and with constraints of arity at most three whose constraint graphs have pathwidth three and from which all sequences of improving flips are exponentially long. They conjectured that for binary variables and constraint graphs of treewidth two, there is always a sequence of improving flips of polynomial length 
from any initial assignment. Our construction shows that no such statement holds once the domain has five values: there, a constraint graph that is a path already admits the all-exp property.

For Max-Cut in graphs of maximum degree~$d$, which is a special case of $(2,d,2)$-MCA (see Section~\ref{sec:notation}), the complexity of local search with the flip neighborhood is well studied. Schäffer and Yannakakis~\cite{SY1991} proved that Max-Cut with the flip neighborhood is PLS-complete under tight reductions and hence has the all-exp property. Poljak~\cite{Pol1995} showed that on graphs of maximum degree three, every sequence of improving flips has length $O(n^2)$, where $n$ is the number of vertices, and asked whether the all-exp property holds for maximum degree four. Monien and Tscheuschner~\cite{MT2010} answered this question affirmatively, and Michel and Scott~\cite{MS2023} later provided a considerably simpler construction. Elsässer and Tscheuschner~\cite{ET2011} proved PLS-completeness for graphs of maximum degree five via a tight reduction; for maximum degree four, PLS-completeness is still open.

\subsection{Our contributions}

As our main result, we prove that $(2,2,r)$-MCA with the flip neighborhood has the all-exp property for $r\ge 5$ (\cref{thm:225} above). 
Our instances for $r=5$ have $2n+1$ variables, $2n+1$ constraints with nonnegative integer values, and a constraint graph that is a path; from the initial assignment, the unique maximal sequence of improving flips has length $2^{n+3}-7$. 
In fact, at every assignment along this sequence except the last one, exactly one flip is improving. 
We present this proof in \cref{sec:225}.

In addition, in Section~\ref{sec:222}, we provide a simpler proof of the following result.
\begin{restatable}{theorem}{Theoremttt}
\label{thm:222-short}
Every strictly improving flip sequence for a $(2,2,2)$-MCA instance on $n$ variables has length $O(n^2)$. 
In particular, $(2,2,2)$-MCA does not have the is-exp property, and hence not the all-exp property.
\end{restatable}

The quadratic length bound for $(2,2,2)$-MCA follows from a more general result 
proven in~\cite{Kaz2020,KCJ2020}. 
We give a more direct and much shorter proof of this statement.

From the two results above, the remaining cases are $r=3$ and $r=4$.
The is-exp property is known for $r=3$~\cite{KCJ2020}, and hence also for $r=4$; whether the all-exp property holds is open in both cases.


\section{Preliminaries}
\label{sec:notation}

An instance of $(p,q,r)$-MCA, the Maximum Constraint Assignment problem with parameters $p$, $q$, and $r$, consists of variables with a common domain of $r$ values and a collection of constraints. Each constraint involves at most $p$ variables and assigns a weight to every combination of values of these variables. The parameter $q$ bounds the number of constraints in which each variable occurs. An assignment of values to all variables thus receives one weight from each constraint, and its objective value, the sum of these weights, is to be maximized.

In the terminology of constraint satisfaction, $(p,q,r)$-MCA is a weighted MAX-CSP over a domain of size $r$, restricted to instances in which every constraint has arity at most $p$ and every variable occurs in at most $q$ constraints.  The constraints are arbitrary: each is given by an explicit table of weights and need not be derived from a predicate of a particular type. 
A weighted constraint that is either satisfied or violated is the special case in which this table takes only two values, the weight of the constraint and zero. Problems of this form are also studied under the name \emph{valued constraint satisfaction problems} (VCSPs)~\cite{SFV1995,CCJK2006,Zivny2012}, which are usually formulated as minimization problems. Since all weights in MCA are finite, $(p,q,r)$-MCA corresponds, after negating the weights, to
finite-valued VCSPs~\cite{TZ2016} of arity at most $p$ in which every variable occurs in at most $q$ constraints. This terminology is also used in the work of Kaznatcheev et al.~\cite{Kaz2020,KCJ2020} on local search, discussed above. We keep the name MCA used by Dumrauf and Monien~\cite{DM2013}, since our results build on their work and since the parameters $p$, $q$, and $r$ make the three restrictions explicit. 

Many classical optimization problems, such as weighted MAX-SAT and Max-Cut, are special cases of MCA. In an instance of Max-Cut, the vertices of the graph become binary variables, and an edge of weight $w$ becomes a constraint on its two endpoints that contributes $w$ if the endpoints receive different values and $0$ otherwise. 

We consider local search with the \emph{flip neighborhood}: two assignments are neighbors if they differ in the value of exactly one variable, and moving to such a neighbor is called a \emph{flip}. The name comes from the binary case, where the only possible change is to flip a variable to its other value; for $r>2$, a flip may assign any of the other $r-1$ values. A flip is \emph{improving} if it strictly increases the objective value, and an assignment is a \emph{local optimum} if no improving flip exists.

\section{The all-exp property for (2,2,5)-MCA}
\label{sec:225}

In this section, we prove our main result:

\Theoremallexp*

It suffices to prove this result for $r=5$. Indeed, an instance of $(2,2,5)$-MCA with nonnegative constraint values can be turned into an instance of $(2,2,r)$-MCA for any $r>5$ by adding $r-5$ new values to the domain of every variable and assigning the value $0$ to every entry of a constraint that involves a new value. Since all constraint values are nonnegative, a flip to a new value never increases the objective value, so the improving sequences starting from an assignment without new values do not change. Hence, it is enough to prove the following statement:

\begin{theorem}\label{thm:225-all-exp}
$(2,2,5)$-MCA with the flip neighborhood has the all-exp property.
\end{theorem}

\begin{proof}
We will prove this result by showing that for every $n\geq1$ there is a $(2,2,5)$-MCA instance $Q_n$ with $2n+1$ variables and nonnegative integer constraint values, and an initial assignment such that every maximal sequence of
strictly improving flips starting from this assignment has length exactly $2^{n+3}-7$.  
Moreover, we will see that at every assignment occurring in this sequence, except the last one, there is a unique strictly improving flip, and that it increases the objective value by exactly one.

The $2n+1$ variables of $Q_n$ are as follows: for $i\in\{1,\ldots,n\}$ we have two variables $a_i,b_i$, and we additionally have a variable $b_0$. 
All variables have the domain $D=\{0,1,2,3,4\}$.

Our instance $Q_n$ contains one unary constraint $K_0$ on the variable $b_0$ and, for every $i \in \{1, \ldots, n\}$, two binary constraints $J_i$ and $K_i$. The constraint $J_i$ contains the two variables $b_{i-1}$ and $a_i$, while the constraint $K_i$ contains the two variables $a_i$ and $b_i$.
Hence, the constraint graph is a path $(b_0, a_1, b_1, \dots, a_n, b_n)$.
Thus, every constraint contains at most two variables, and every variable occurs in at most two constraints, so the resulting instance $Q_n$ is a $(2,2,5)$-MCA instance.

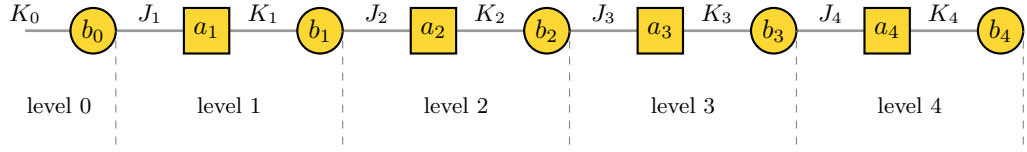
\begin{figure}
\centering 

\begin{tikzpicture}[scale=0.75, var/.style={draw=black,thick,fill=varyellow}]

  \draw[gray!80,line width=1pt] (-1.2,0) -- (16,0);

  \foreach \x/\n in {
      -1.2/K_0,
       1/J_1,  3/K_1,
       5/J_2,  7/K_2,
       9/J_3, 11/K_3,
      13/J_4, 15/K_4}{
        \node[] at (\x,0.3) {\small $\n$};
  }

  \foreach \k in {0,...,4}{
    \filldraw[var] (4*\k,0) circle (0.4);
    \node at (4*\k,0) {$b_{\k}$};
  }
  \foreach \k in {1,...,4}{
    \filldraw[var] ({4*\k-2.4},-0.4) rectangle ({4*\k-1.6},0.4);
    \node at ({4*\k-2},0) {$a_{\k}$}; 
  }

  \node[font=\footnotesize,anchor=north] at (-0.6,-1) {level 0};
  \foreach \i in {1,...,4}{
     \node[font=\footnotesize,anchor=north] at ({4*\i-1.6},-1) {level \i};
  }
  \foreach \i in {0,...,4}{
    \draw[gray, dashed] ({4*\i+0.4},-2) -- ({4*\i+0.4},0);
  }
\end{tikzpicture}
\caption{Visualization of the instance $Q_4$ showing its path-like structure. The variables $b_i$ are shown as circles, the variables $a_i$ as squares. An edge between two variables represents a constraint containing these two variables. The constraint $K_0$ only contains the variable $b_0$.}
\label{fig:Q4}
\end{figure}

We group the variables and constraints into \emph{levels}: level $0$ consists of the variable $b_0$, and for $1\leq i\leq n$, level $i$ consists of the two variables $a_i$ and $b_i$.  The constraint $K_0$ is the constraint of level $0$, and $J_i$ and $K_i$ are the constraints of level $i$; $K_i$ involves only the variables of level $i$, whereas $J_i$ connects level $i$ to level $i-1$.  Figure~\ref{fig:Q4} shows a visualization of the instance $Q_4$. For $0\leq i\leq n$ let
$V_i:=\{b_0,a_1,b_1,\ldots,a_i,b_i\}$
be the set of variables of levels $0,\ldots,i$.  When level $i$ is considered, we refer to the variables in $V_{i-1}$ as the \emph{lower levels}.

For readability, we denote the five values of an $a_i$-variable by
\begin{equation*}
        \avalzero=0,\qquad \avalone=1,\qquad \avaltwo=2,\qquad \avalthree=3,\qquad \avalfour=4.  
\end{equation*}

The constraint values depend on two sequences of integer parameters 
\begin{equation*}
    H_i := 16\cdot2^i-14\qquad\text{ and} \qquad M_i :=100^{i+1}  \qquad \text{ for } i=0, \ldots, n.
\end{equation*}
Note that all $H_i$ are even, that $H_i=2H_{i-1}+14$ for $i\geq1$, and that $H_{i-1}<H_i < M_i$ for all $i\geq1$.

The values of the constraints $K_0$, $J_i$, and $K_i$ are given in Tables~\ref{tab:K0}, \ref{tab:Ji}, and \ref{tab:Ki}.

\begin{table}[t!]
\centering 
\caption{Definition of the values for the constraint $K_0$ depending on the value of the variable $b_0$.}
\label{tab:K0}
$
\begin{array}{c|ccccc}
b_0&0&1&2&3&4\\ \hline
K_0(b_0)&100&101&0&0&0
\end{array}
$
\end{table}

\begin{table}[t!]
\centering 
\caption{Definition of the values for the constraints $J_i$ for $i=1,\ldots, n$ depending on the values of the variables $b_{i-1}$ and $a_i$.}
\label{tab:Ji}
$
\begin{array}{c|ccccc}
J_i(b_{i-1},a_i)&a_i=\avalzero&a_i=\avalone&a_i=\avaltwo&a_i=\avalthree&a_i=\avalfour\\ \hline
b_{i-1}=0&2M_i+H_{i-1}&M_i&0&M_i+H_{i-1}&2M_i+H_{i-1}\\
b_{i-1}=1&0&M_i&2M_i&M_i&0\\
b_{i-1}=2&\tfrac12H_{i-1}&M_i&0&M_i+\tfrac12H_{i-1}&\tfrac12H_{i-1}\\
b_{i-1}=3&0&M_i&0&M_i&0\\
b_{i-1}=4&H_{i-1}&M_i&0&M_i+H_{i-1}&H_{i-1}
\end{array}
$
\end{table}

\begin{table}[t!]
\centering 
\caption{Definition of the values for the constraints $K_i$ for $i=1,\ldots, n$  depending on the values of the variables $a_i$ and $b_i$.}
\label{tab:Ki}
$
\begin{array}{c|ccccc}
K_i(a_i,b_i)&b_i=0&b_i=1&b_i=2&b_i=3&b_i=4\\ \hline
a_i=\avalzero&9M_i&0&9M_i+1&0&9M_i-1\\
a_i=\avalone&0&0&10M_i+H_{i-1}+2&0&0\\
a_i=\avaltwo&0&0&9M_i+H_{i-1}+3&9M_i+H_{i-1}+4&9M_i-3\\
a_i=\avalthree&0&0&0&10M_i+H_{i-1}+5&10M_i-2\\
a_i=\avalfour&0&9M_i+H_{i-1}+7&9M_i+1&9M_i+H_{i-1}+6&0
\end{array}
$
\end{table}

The initial assignment assigns $0$ to all $b_i$-variables and $\avalzero$ to all $a_i$-variables; since $\avalzero=0$, it is the all-zero assignment.

We write $\flip{x}{u}{v}$ for the flip that changes the value of the variable $x$ from $u$ to $v$. 
Recall that a flip is improving if it strictly increases the objective value; we call the increase its \emph{gain}. Since every variable occurs in at most two constraints, the gain of a flip of a variable depends only on the values of that variable and of its (at most two) neighbors on the path.  Explicitly:
\begin{itemize}
\item the gain of a flip of $a_i$ is determined by $b_{i-1}$ and
      $b_i$, through $J_i$ and $K_i$;
\item the gain of a flip of $b_i$, $1\leq i<n$, is determined by
      $a_i$ and $a_{i+1}$, through $K_i$ and $J_{i+1}$;
\item the gain of a flip of $b_n$ is determined by $a_n$, through
      $K_n$ alone;
\item the gain of a flip of $b_0$ is determined by $a_1$, through
      $K_0$ and $J_1$.
\end{itemize}

By the \emph{state} of level $i\geq1$ we mean the pair of values $(x,y)$ currently assigned to $(a_i,b_i)$; we write ``level $i$ is at $(x,y)$'' or simply ``at $(x,y)$''.  Let 
\begin{equation*}
  \mathcal A=\bigl\{(\avalzero,0),(\avalzero,2),(\avalzero,4),(\avalone,2),(\avaltwo,2),(\avaltwo,3),(\avaltwo,4),
                   (\avalthree,3),(\avalthree,4),(\avalfour,1),(\avalfour,2),(\avalfour,3)\bigr\}
\end{equation*}
be the set of pairs of values at which the table of $K_i$ has a
nonzero entry; this set is the same for every $i$.  We call the pairs in
 $\mathcal A$ \emph{admissible}.
Every entry $K_i(x,y)$ with $(x,y)\in\mathcal A$ is at least
$9M_i-3$, whereas $K_i(x,y)=0$ for $(x,y)\notin\mathcal A$.  A flip of
$a_i$ or of $b_i$ that changes the state of level $i$ from an
admissible to a non-admissible pair therefore loses at least $9M_i-3$
in $K_i$.

In the following three claims, we analyze those combinations of states and variable values that occur in the improving sequence described afterwards.
Note that we are only concerned with admissable pairs, since as argued above, starting from admissiable pairs, we will not encounter inadmissible pairs in an improving flip sequence.

\begin{claim}[Flips of $a_i$]\label{cl:a}
Let $1\leq i\leq n$.  For each of the following combinations of the
state of level $i$ and the value of $b_{i-1}$, the improving flips of
$a_i$ are as stated, and each flip listed is the only improving flip
of $a_i$ and has gain $1$.

\begin{enumerate}[label=(\alph*),ref=\theclaim(\alph*)]
\item\label{cl:a:idle}
      If level $i$ is at $(\avalzero,0)$ or at $(\avalfour,1)$, then for every value
      of $b_{i-1}$ the variable $a_i$ has no improving flip.
\item\label{cl:a:F}
      If level $i$ is at $(\avalone,2)$, then $a_i$ has an improving flip if
      and only if $b_{i-1}=1$, namely $\flip{a_i}{\avalone}{\avaltwo}$.
\item\label{cl:a:H3}
      If level $i$ is at $(\avalthree,3)$, then $a_i$ has an improving flip if
      and only if $b_{i-1}=0$, namely $\flip{a_i}{\avalthree}{\avalfour}$.
\item\label{cl:a:H4}
      If level $i$ is at $(\avalthree,4)$, then $a_i$ has an improving flip if
      and only if $b_{i-1}=0$, namely $\flip{a_i}{\avalthree}{\avalzero}$.
\item\label{cl:a:zero}
      Let $b_{i-1}=0$.  If level $i$ is at $(\avalzero,2)$ or at $(\avalfour,2)$, then
      $a_i$ has the improving flip $\flip{a_i}{\avalzero}{\avalone}$ or
      $\flip{a_i}{\avalfour}{\avalone}$, respectively.  If level $i$ is at $(\avalfour,3)$ or
      at $(\avalzero,4)$, then $a_i$ has no improving flip.
\item\label{cl:a:one}
      Let $b_{i-1}=1$.  If level $i$ is at $(\avaltwo,3)$ or at $(\avaltwo,4)$, then
      $a_i$ has the improving flip $\flip{a_i}{\avaltwo}{\avalthree}$.  If level $i$ is
      at $(\avaltwo,2)$, then $a_i$ has no improving flip.
\end{enumerate}
\end{claim}

\begin{proof}[Proof of Claim~\ref{cl:a}]
A flip of $a_i$ changes only the constraints $J_i$ and $K_i$.  For
fixed values of $b_{i-1}$ and $b_i$ and a value $y\in D$ of $a_i$ let
\begin{equation*}
        \Phi(y) := J_i(b_{i-1},y)+K_i(y,b_i).   
\end{equation*}
Then the flip $\flip{a_i}{x}{y}$ is improving if and only if
$\Phi(y)>\Phi(x)$, and its gain is $\Phi(y)-\Phi(x)$. Since
$H_{i-1}$ is even, all values of $\Phi$ are integers.

From any of the admissible states, a flip to a non-admissible pair loses at least $9M_i-3$ in $K_i$ and
gains at most $2M_i+H_{i-1}$ in $J_i$; since $H_{i-1}<M_i$, it is not
improving.  

It therefore suffices to compare the admissible alternatives.
Using Tables~\ref{tab:Ji} and~\ref{tab:Ki}, we obtain:
\[
\begin{array}{c | c  | c | c}
(a_i,b_i)
    & b_{i-1}
    & \text{improving target for $a_i$}
    & \text{gain}\\
\hline
(\avalzero,0)  & \text{any} & \text{none} & -\\
(\avalfour,1)  & \text{any} & \text{none} & -\\
\hline
(\avalone,2)   & 1          & \avaltwo    & 1\\
(\avalone,2)   & \neq 1     & \text{none} & -\\
\hline
(\avalthree,3) & 0          & \avalfour   & 1\\
(\avalthree,3) & \neq 0     & \text{none} & -\\
(\avalthree,4) & 0          & \avalzero   & 1\\
(\avalthree,4) & \neq 0     & \text{none} & -\\
\hline
(\avalzero,2)  & 0          & \avalone    & 1\\
(\avalfour,2)  & 0          & \avalone    & 1\\
(\avalfour,3)  & 0          & \text{none} & -\\
(\avalzero,4)  & 0          & \text{none} & -\\
\hline
(\avaltwo,3)   & 1          & \avalthree  & 1\\
(\avaltwo,4)   & 1          & \avalthree  & 1\\
(\avaltwo,2)   & 1          & \text{none} & -
\end{array}
\]
Here ``none'' means that every admissible alternative has value at most
$\Phi(a_i)$. The first two rows are immediate because
$(\avalzero,0)$ and $(\avalfour,1)$ are the only admissible pairs with
second component $0$ and $1$, respectively. The remaining entries follow
by direct comparison of the corresponding entries in
Tables~\ref{tab:Ji} and~\ref{tab:Ki}. This proves all parts of the claim.
\end{proof}

\begin{claim}[Flips of $b_i$]\label{cl:b}
Let $1\leq i\leq n$.  For each of the following combinations of the
state of level $i$ and, if $i<n$, the value of $a_{i+1}$, the
improving flips of $b_i$ are as stated, and each flip listed is the
only improving flip of $b_i$ and has gain $1$.
\begin{enumerate}[label=(\alph*),ref=\theclaim(\alph*)]
\item\label{cl:b:F}
      Let $i=n$ or $a_{i+1}=\avalone$.  If level $i$ is at $(\avalzero,0)$, $(\avaltwo,2)$,
      or $(\avalfour,3)$, then $b_i$ has the improving flip $\flip{b_i}{0}{2}$,
      $\flip{b_i}{2}{3}$, or $\flip{b_i}{3}{1}$, respectively.  If level $i$
      is at $(\avalzero,2)$, $(\avalone,2)$, $(\avaltwo,3)$, $(\avalthree,3)$, or $(\avalfour,1)$, then $b_i$
      has no improving flip.
\item\label{cl:b:H}
      Let $i<n$ and $a_{i+1}=\avalthree$.  If level $i$ is at $(\avalfour,1)$, $(\avaltwo,2)$,
      or $(\avalzero,4)$, then $b_i$ has the improving flip $\flip{b_i}{1}{2}$,
      $\flip{b_i}{2}{4}$, or $\flip{b_i}{4}{0}$, respectively.  If level $i$
      is at $(\avalfour,2)$, $(\avalone,2)$, $(\avaltwo,4)$, $(\avalthree,4)$, or $(\avalzero,0)$, then $b_i$
      has no improving flip.
\item\label{cl:b:pin}
      Let $i<n$.  If level $i$ is at $(\avalzero,0)$ and $a_{i+1}\in\{\avalzero,\avalfour\}$,
      or if level $i$ is at $(\avalfour,1)$ and $a_{i+1}=\avaltwo$, then $b_i$ has no
      improving flip.
\end{enumerate}
\end{claim}

\begin{proof}[Proof of Claim~\ref{cl:b}]
A flip of $b_i$ changes only the constraint $K_i$ and, if $i<n$, the
constraint $J_{i+1}$.  For fixed values of $a_i$ and, if $i<n$, of
$a_{i+1}$, and a value $y\in D$ of $b_i$ let
\begin{equation*}
        \Psi(y) :=
        \begin{cases}
          K_i(a_i,y)+J_{i+1}(y,a_{i+1}), & i<n,\\
          K_i(a_i,y), & i=n.
        \end{cases}
\end{equation*}
Then the flip $\flip{b_i}{x}{y}$ is improving if and only if
$\Psi(y)>\Psi(x)$, and its gain is $\Psi(y)-\Psi(x)$.  The values of
$J_{i+1}$ are given by Table~\ref{tab:Ji} with $i$ replaced by $i+1$;
its entries involve $M_{i+1}$ and $H_i=2H_{i-1}+14$.


Let $i=n$ or $a_{i+1}=\avalone$. Since
$J_{i+1}(y,\avalone)=M_{i+1}$ for all $y\in D$, in both cases
$\Psi(y)$ differs from $K_i(a_i,y)$ by a constant. The relevant
comparisons are therefore:
\[
\begin{array}{c|c}
a_i & \text{relevant ordering of values}\\
\hline
\avalzero
    & \Psi(2)=\Psi(0)+1>\Psi(y)
      \quad (y\notin\{0,2\})\\
\avalone
    & \Psi(2)>\Psi(y)
      \quad (y\neq2)\\
\avaltwo
    & \Psi(3)=\Psi(2)+1>\Psi(y)
      \quad (y\notin\{2,3\})\\
\avalthree
    & \Psi(3)>\Psi(y)
      \quad (y\neq3)\\
\avalfour
    & \Psi(1)=\Psi(3)+1>\Psi(y)
      \quad (y\notin\{1,3\})
\end{array}
\]
Hence, at $(\avalzero,0)$, $(\avaltwo,2)$, and
$(\avalfour,3)$ the only improving flips are
$\flip{b_i}{0}{2}$, $\flip{b_i}{2}{3}$, and
$\flip{b_i}{3}{1}$, respectively, and each has gain $1$.
At $(\avalzero,2)$, $(\avalone,2)$, $(\avaltwo,3)$,
$(\avalthree,3)$, and $(\avalfour,1)$ there is no improving
flip. This proves Claim~\ref{cl:b:F}.

Let $i<n$ and $a_{i+1}=\avalthree$.  By Table~\ref{tab:Ji}, the values
$J_{i+1}(y,\avalthree)-M_{i+1}$ for $y=0,1,2,3,4$ are
$H_i,\,0,\,\tfrac12H_i,\,0,\,H_i$.
Adding the entries of the table for $K_i$, and noting that $H_{i-1}=\frac{1}{2}H_{i}-7$, we obtain the following values of $\Psi(y)-M_{i+1}$:
\[
\begingroup
\setlength{\arraycolsep}{4pt}
\begin{array}{c|ccccc}
\Psi(y)-M_{i+1}
    &y=0&y=1&y=2&y=3&y=4\\ \hline
a_i = \avalzero
    &9M_i+H_i
    &0
    &9M_i+\frac12H_i+1
    &0
    &9M_i+H_i-1\\
a_i = \avalone
    &H_i
    &0
    &10M_i+H_i-5
    &0
    &H_i\\
a_i = \avaltwo
    &H_i
    &0
    &9M_i+H_i-4
    &9M_i+\frac12H_i-3
    &9M_i+H_i-3\\
a_i = \avalthree
    &H_i
    &0
    &\frac12H_i
    &10M_i+\frac12H_i-2
    &10M_i+H_i-2\\
a_i = \avalfour
    &H_i
    &9M_i+\frac12H_i
    &9M_i+\frac12H_i+1
    &9M_i+\frac12H_i-1
    &H_i
\end{array}
\endgroup
\]

Since $H_i < M_i$, we then obtain the following comparisons:
\[
\begin{array}{c|c}
a_i & \text{relevant ordering of values}\\
\hline
\avalzero
    & \Psi(0)=\Psi(4)+1>\Psi(y)
      \quad (y\notin\{0,4\})\\
\avalone
    & \Psi(2)>\Psi(y)
      \quad (y\neq2)\\
\avaltwo
    & \Psi(4)=\Psi(2)+1>\Psi(y)
      \quad (y\notin\{2,4\})\\
\avalthree
    & \Psi(4)>\Psi(y)
      \quad (y\neq4)\\
\avalfour
    & \Psi(2)=\Psi(1)+1>\Psi(y)
      \quad (y\notin\{1,2\})
\end{array}
\]
Claim~\ref{cl:b:H} then immediately follows.

Let $i<n$.  If level $i$ is at $(\avalzero,0)$ and $a_{i+1}\in\{\avalzero,\avalfour\}$, then
by Table~\ref{tab:Ji} we have $J_{i+1}(0,a_{i+1})=2M_{i+1}+H_i$ and
$J_{i+1}(y,a_{i+1})\leq H_i$ for $y\neq0$.  Hence every flip of $b_i$
loses at least $2M_{i+1}$ in $J_{i+1}$, while it gains at most $1$ in
$K_i$.  If level $i$ is at $(\avalfour,1)$ and $a_{i+1}=\avaltwo$, then
$J_{i+1}(1,\avaltwo)=2M_{i+1}$ and $J_{i+1}(y,\avaltwo)=0$ for $y\neq1$.  Hence
every flip of $b_i$ loses $2M_{i+1}$ in $J_{i+1}$, while $K_i(\avalfour,1)$ is
the largest entry of row $\avalfour$.  In both cases no flip of $b_i$ is
improving.  This proves Claim~\ref{cl:b:pin}.
\end{proof}

\begin{claim}[Flips of $b_0$]\label{cl:b0}
Each flip listed below is the only improving flip of $b_0$ and has
gain $1$.
\begin{enumerate}[label=(\alph*),ref=\theclaim(\alph*)]
\item\label{cl:b0:F}
      If $a_1=\avalone$, then at $b_0=0$ the variable $b_0$ has the improving
      flip $\flip{b_0}{0}{1}$, and at $b_0=1$ it has no improving flip.
\item\label{cl:b0:H}
      If $a_1=\avalthree$, then at $b_0=1$ the variable $b_0$ has the improving
      flip $\flip{b_0}{1}{0}$, and at $b_0=0$ it has no improving flip.
\item\label{cl:b0:pin}
      If $a_1\in\{\avalzero,\avalfour\}$ and $b_0=0$, or if $a_1=\avaltwo$ and $b_0=1$, then
      $b_0$ has no improving flip.
\end{enumerate}
\end{claim}

\begin{proof}[Proof of Claim~\ref{cl:b0}]
A flip of $b_0$ changes only $K_0(b_0)+J_1(b_0,a_1)$.  By
Table~\ref{tab:Ji} with $i=1$, where $H_0=2$, the values of this sum
for $b_0=0,1,2,3,4$ are $M_1+(100,101,0,0,0)$ if $a_1=\avalone$, which
proves Claim~\ref{cl:b0:F}, and $M_1+(102,101,1,0,2)$ if $a_1=\avalthree$,
which proves Claim~\ref{cl:b0:H}.  If $a_1\in\{\avalzero,\avalfour\}$, they are
$(102+2M_1,\,101,\,1,\,0,\,2)$, and if $a_1=\avaltwo$, they are
$(100,\,101+2M_1,\,0,\,0,\,0)$; in both cases the value at the current
value of $b_0$ exceeds every other value by more than $2M_1$.  This
proves Claim~\ref{cl:b0:pin}.
\end{proof}

\subparagraph{The recursive flip sequences.}
We now describe a sequence of flips and show afterwards that it is the unique maximal improving sequence starting at the initial assignment.

For $0\leq i\leq n$ we define two assignments of $V_i$:
\begin{itemize}
\item $\mathbf G_i$ assigns the value $0$ to every $b$-variable and the value $\avalzero$ to every $a$-variable in $V_i$; since $\avalzero=0$, this is the all-zero assignment of $V_i$;
\item $\mathbf R_i$ coincides with $\mathbf G_i$, except that
      $b_i=1$ and, if $i\geq1$, $a_i=\avalfour$.
\end{itemize}
For $i\geq1$, both restrict to $\mathbf G_{i-1}$ on $V_{i-1}$; level
$i$ is at $(\avalzero,0)$ in $\mathbf G_i$ and at $(\avalfour,1)$ in $\mathbf R_i$.

We define two sequences of flips of variables in $V_i$ recursively:
$\alpha_0=(\flip{b_0}{0}{1})$, $\beta_0=(\flip{b_0}{1}{0})$, and for
$i\geq1$
\[
\begin{aligned}
\alpha_i=\bigl(&\flip{b_i}{0}{2},\ \flip{a_i}{\avalzero}{\avalone},\ \alpha_{i-1},\
  \flip{a_i}{\avalone}{\avaltwo},\ \flip{b_i}{2}{3},\ \flip{a_i}{\avaltwo}{\avalthree},\ \\ 
  & \qquad \beta_{i-1},\
  \flip{a_i}{\avalthree}{\avalfour},\ \flip{b_i}{3}{1}\bigr),\\
\beta_i=\bigl(&\flip{b_i}{1}{2},\ \flip{a_i}{\avalfour}{\avalone},\ \alpha_{i-1},\
  \flip{a_i}{\avalone}{\avaltwo},\ \flip{b_i}{2}{4},\ \flip{a_i}{\avaltwo}{\avalthree},\ \\
  & \qquad \beta_{i-1},\
  \flip{a_i}{\avalthree}{\avalzero},\ \flip{b_i}{4}{0}\bigr).
\end{aligned}
\]
By induction on $i$, $\alpha_i$ transforms $\mathbf G_i$ into
$\mathbf R_i$ and $\beta_i$ transforms $\mathbf R_i$ into
$\mathbf G_i$.

\subparagraph{Length of the sequences.}
The concatenation $\alpha_i\beta_i$ is one full cycle of level $i$.
During it, $a_i$ runs once through the values
\[
        \avalzero,\ \avalone,\ \avaltwo,\ \avalthree,\ \avalfour,\ \avalone,\ \avaltwo,\ \avalthree
\]
and returns to $\avalzero$, while $b_i$ runs through $0,2,3,1,2,4$ and
returns to $0$.  The lower levels perform $\alpha_{i-1}$ in each of
the two phases with $a_i=\avalone$, and $\beta_{i-1}$ in each of the two
phases with $a_i=\avalthree$; that is, they complete two full cycles of level
$i-1$ during one full cycle of level $i$.  This doubling is the source
of the exponential length.  Precisely, for $i\geq1$ each of $\alpha_i$ and $\beta_i$ consists of seven flips of level $i$, one copy of $\alpha_{i-1}$, and one copy of $\beta_{i-1}$. Hence, by induction on $i$, the sequences $\alpha_i$ and $\beta_i$ have the same length $L_i$, where $L_0=1$ and $L_i=2L_{i-1}+7$ for $i\geq1$. Consequently,
\[
        L_i=2^{i+3}-7.
\]

We also note the following consequence of the definitions:
\begin{itemize}
    \item [(*)] along $\alpha_i$, the variable $b_i$ has the value $1$ only at the
final assignment $\mathbf R_i$; along $\beta_i$, it has the value $0$
only at the final assignment $\mathbf G_i$.
\end{itemize}

Indeed, along $\alpha_i$ the variable $b_i$ runs through $0,2,3,1$,
and along $\beta_i$ through $1,2,4,0$ (for $i=0$ through $0,1$ and
$1,0$).

\subparagraph{Uniqueness of the improving sequence.}
We now show that $\alpha_n$ is the unique maximal improving sequence starting at the initial assignment. 

Since the gains of variables in $V_i$ depend only on the values of
the variables in $V_i\cup\{a_{i+1}\}$, where $a_{n+1}$ is omitted for $i=n$, the following statements are
well defined.  We prove them simultaneously by induction on $i$.

\medskip
\noindent
$(A_i)$, for $0\leq i\leq n$:
Assume that $a_{i+1}=\avalone$ if $i<n$.  Start at an assignment that
restricts to $\mathbf G_i$ on $V_i$ and perform $\alpha_i$. Then at
every assignment before the last one, the next flip of $\alpha_i$ is
the unique improving flip of a variable in $V_i$, and it has gain $1$;
at the last assignment, which restricts to $\mathbf R_i$ on $V_i$, no variable of $V_i$
has an improving flip.

\medskip
\noindent
$(B_i)$, for $0\leq i<n$:
Assume that $a_{i+1}=\avalthree$. Start at an assignment that
restricts to $\mathbf R_i$ on $V_i$ and perform $\beta_i$. Then at
every assignment before the last one, the next flip of $\beta_i$ is
the unique improving flip of a variable in $V_i$, and it has gain $1$;
at the last assignment, which restricts to $\mathbf G_i$ on $V_i$,
no variable of $V_i$ has an improving flip.

\medskip
In other words, $(B_i)$ is the same as $(A_i)$, except that we assume $a_{i+1}=\avalthree$, start at $\mathbf R_i$, and perform
$\beta_i$, which ends at $\mathbf G_i$.
We now derive a consequence that we use repeatedly.

\medskip
\noindent
$(P_i)$, for $0\leq i<n$:
If the variables of $V_i$ are in $\mathbf G_i$ and $a_{i+1}\in\{\avalzero,\avalfour\}$,
or if they are in $\mathbf R_i$ and $a_{i+1}=\avaltwo$, then no variable of
$V_i$ has an improving flip.

\medskip
\noindent
Indeed, assume $(A_i)$ and $(B_i)$.  By $(A_i)$, applied at its first
assignment, the only improving flip in $V_i$ at $\mathbf G_i$ with
$a_{i+1}=\avalone$ is the first flip of $\alpha_i$, which is a flip of
$b_i$.  The gains of the variables in $V_i\setminus\{b_i\}$ do not
depend on $a_{i+1}$, so none of them has an improving flip at
$\mathbf G_i$, regardless of the value of $a_{i+1}$.  The same argument
with $(B_i)$ applies at $\mathbf R_i$, since the first flip of
$\beta_i$ is also a flip of $b_i$.  The remaining variable $b_i$ has no improving flip in both situations of $(P_i)$, by
Claim~\ref{cl:b:pin} if $i\geq1$ and by Claim~\ref{cl:b0:pin} if $i=0$.

\medskip
\noindent
\textbf{Base case $i=0$.}
Here $V_0=\{b_0\}$, and $(A_0)$ and $(B_0)$ are Claims~\ref{cl:b0:F} and~\ref{cl:b0:H}.

\medskip
\noindent
\textbf{Induction step.}
Let $1\leq i\leq n$ and assume $(A_{i-1})$ and $(B_{i-1})$, and hence
$(P_{i-1})$.  We prove $(A_i)$; for this, we assume $i=n$ or $a_{i+1}=\avalone$.  The
assignments occurring along $\alpha_i$ are listed in the following
table, together with the state of the lower levels $V_{i-1}$ and the
improving flip.
\[
\begin{array}{c|c|l}
(a_i,b_i)&\text{lower levels}&\text{improving flip}\\ \hline
(\avalzero,0)&\mathbf G_{i-1}&\flip{b_i}{0}{2}\\
(\avalzero,2)&\mathbf G_{i-1}&\flip{a_i}{\avalzero}{\avalone}\\
(\avalone,2)&\alpha_{i-1}\text{ is performed}
      &\text{the flips of }\alpha_{i-1},\text{ then }\flip{a_i}{\avalone}{\avaltwo}\\
(\avaltwo,2)&\mathbf R_{i-1}&\flip{b_i}{2}{3}\\
(\avaltwo,3)&\mathbf R_{i-1}&\flip{a_i}{\avaltwo}{\avalthree}\\
(\avalthree,3)&\beta_{i-1}\text{ is performed}
      &\text{the flips of }\beta_{i-1},\text{ then }\flip{a_i}{\avalthree}{\avalfour}\\
(\avalfour,3)&\mathbf G_{i-1}&\flip{b_i}{3}{1}\\
(\avalfour,1)&\mathbf G_{i-1}&\text{none}
\end{array}
\]
We verify each row; in every case, all variables of $V_i$ other than
the one indicated have no improving flip, and the indicated flip has
gain $1$.

\emph{Rows with the lower levels in $\mathbf G_{i-1}$ or $\mathbf R_{i-1}$.}
In the rows $(\avalzero,0)$, $(\avalzero,2)$, $(\avalfour,3)$, and $(\avalfour,1)$, the lower levels are in
$\mathbf G_{i-1}$ and $a_i\in\{\avalzero,\avalfour\}$; in the rows $(\avaltwo,2)$ and $(\avaltwo,3)$,
they are in $\mathbf R_{i-1}$ and $a_i=\avaltwo$.  By $(P_{i-1})$, no
variable of $V_{i-1}$ has an improving flip.  The flips of $a_i$ are
given by Claims~\ref{cl:a:idle} and~\ref{cl:a:zero} with $b_{i-1}=0$ in the
rows $(\avalzero,0)$, $(\avalzero,2)$, $(\avalfour,3)$, $(\avalfour,1)$, and by Claim~\ref{cl:a:one} with
$b_{i-1}=1$ in the rows $(\avaltwo,2)$, $(\avaltwo,3)$.  The flips of $b_i$ are
given by Claim~\ref{cl:b:F}. Together they yield exactly the flips in the table.

\emph{The forward traversal.}
In the row $(\avalone,2)$ we have $a_i=\avalone$, so $(A_{i-1})$ applies to
$V_{i-1}$: at each assignment, the next flip of $\alpha_{i-1}$ is the
unique improving flip in $V_{i-1}$, and at the final assignment
$\mathbf R_{i-1}$ there is none.  By Claim~\ref{cl:b:F}, $b_i$ has
no improving flip.  By Claim~\ref{cl:a:F}, $a_i$ has an improving
flip exactly when $b_{i-1}=1$, which by 
(*) happens only
at $\mathbf R_{i-1}$; the flip is then $\flip{a_i}{\avalone}{\avaltwo}$.

\emph{The reverse traversal.}
In the row $(\avalthree,3)$ we have $a_i=\avalthree$, so $(B_{i-1})$ applies to
$V_{i-1}$.  By Claim~\ref{cl:b:F}, $b_i$ has no improving flip.  By
Claim~\ref{cl:a:H3}, $a_i$ has an improving flip exactly when
$b_{i-1}=0$, which by 
(*) happens only at the final
assignment $\mathbf G_{i-1}$ of $\beta_{i-1}$; the flip is then
$\flip{a_i}{\avalthree}{\avalfour}$.

This proves $(A_i)$. For $i<n$ we prove $(B_i)$ in the same way;
now $a_{i+1}=\avalthree$, and the assignments along $\beta_i$ are
\[
\begin{array}{c|c|l}
(a_i,b_i)&\text{lower levels}&\text{improving flip}\\ \hline
(\avalfour,1)&\mathbf G_{i-1}&\flip{b_i}{1}{2}\\
(\avalfour,2)&\mathbf G_{i-1}&\flip{a_i}{\avalfour}{\avalone}\\
(\avalone,2)&\alpha_{i-1}\text{ is performed}
      &\text{the flips of }\alpha_{i-1},\text{ then }\flip{a_i}{\avalone}{\avaltwo}\\
(\avaltwo,2)&\mathbf R_{i-1}&\flip{b_i}{2}{4}\\
(\avaltwo,4)&\mathbf R_{i-1}&\flip{a_i}{\avaltwo}{\avalthree}\\
(\avalthree,4)&\beta_{i-1}\text{ is performed}
      &\text{the flips of }\beta_{i-1},\text{ then }\flip{a_i}{\avalthree}{\avalzero}\\
(\avalzero,4)&\mathbf G_{i-1}&\flip{b_i}{4}{0}\\
(\avalzero,0)&\mathbf G_{i-1}&\text{none}
\end{array}
\]
The verification is identical, using Claim~\ref{cl:b:H} instead of Claim~\ref{cl:b:F}, and Claim~\ref{cl:a:H4} instead of Claim~\ref{cl:a:H3} in the reverse traversal.  This completes the induction.

We are now able to conclude the proof. 
For $i=n$, we have $V_n=\{b_0,a_1,b_1,\ldots,a_n,b_n\}$, the set of all
variables, and $(A_n)$ holds without any assumptions on further
variables. The assignment $\mathbf G_n$ is the initial assignment.  By $(A_n)$, at every assignment along $\alpha_n$ except the last one, there is exactly one improving flip, namely the next flip of $\alpha_n$, and
the final assignment $\mathbf R_n$ is a local optimum. Hence $\alpha_n$ is the unique maximal strictly improving sequence starting at $\mathbf G_n$, and its length is
\[
        L_n=2^{n+3}-7.
\]
The instance $Q_n$ has $2n+1$ variables and $2n+1$ constraints, and all its constraint values are integers with $O(n)$ bits. However, $L_n$ is exponential in $n$.
This shows that $(2,2,5)$-MCA with the flip neighborhood has the all-exp property.
\end{proof}

\section{(2,2,2)-MCA flip sequences have at most quadratic length}
\label{sec:222}

In this section, we give a direct proof of the at most quadratic length of all $(2,2,2)$-MCA flip sequences by compressing the coefficients without changing which flips are improving. 
The key point is that the constraint graph has maximum degree two. 
Thus, for each variable, the gain of a flip depends on at most two neighboring variables, and preserving the signs and relative magnitudes of the corresponding coefficients is enough to preserve all improving flips. 
We replace the original coefficients by small integers with this property. 
The resulting integer-valued objective has range $O(n^2)$, which bounds the length of every strictly improving sequence. 
The same conclusion, with the better constant $\binom{n+1}{2}$ for tree-structured instances, also follows from the more general results of~\cite{Kaz2020,KCJ2020}.

\Theoremttt*

\begin{proof}
Let $x = (x_1, \ldots, x_n)$ be the $n$ Boolean variables of a $(2,2,2)$-MCA instance and let 
$F$ be the objective function. 
Every function of one or two Boolean variables can be written as a multilinear polynomial of degree at most two. 
Since every constraint involves at most two variables, collecting the terms of all constraints gives
\[
    F(x)=K+\sum_{i=1}^n h_ix_i+\sum_{\{i,j\}\in E}c_{ij}x_ix_j,
\]
where $E$ is the set of pairs with $c_{ij}\neq0$. 

Let $G=(\{1,\dots,n\},E)$. A term $x_ix_j$ can only arise from a constraint containing both $x_i$ and $x_j$. 
Since every variable occurs in at most two constraints, $G$ has maximum degree at most two. 
Hence, writing $m:=|E|$, we have $m\le n$.

For $i\in\{1,\dots,n\}$, let $N(i)$ be its neighbors in $G$ and define $q_i(x):=\sum_{j\in N(i)}c_{ij}x_j$. 
Flipping $x_i$ from its current value to $1-x_i$ changes the objective value by $(1-2x_i)(h_i+q_i(x))$. Thus, to preserve which flips are improving, it is enough to preserve the sign of $h_i+q_i(x)$ for every $i$ and every $x\in\{0,1\}^n$.

We first replace the coefficients $c_{ij}$ by small integers. 
These coefficients are shared by the flips of both endpoints of an edge and therefore have to be replaced consistently. 
Once they have been fixed, the coefficient $h_i$ can be replaced independently for each $i$.

\smallskip
\noindent\textbf{Compressing the quadratic coefficients.}
For an edge $e=ij$, write $c_e:=c_{ij}$. 
Let $0<\alpha_1<\cdots<\alpha_s$ be the distinct values among $\{|c_e|:e\in E\}$. 
For each $e\in E$, let $r_e$ be the index for which $|c_e|=\alpha_{r_e}$, and set $d_e:=2\operatorname{sgn}(c_e)r_e$. 
Thus, the replacement of $c_e$ by $d_e$ preserves the sign of every coefficient, and for any $e,f\in E$ it also preserves whether $|c_e|<|c_f|$, $|c_e|=|c_f|$, or $|c_e|>|c_f|$.
Define $q_i'(x):=\sum_{j\in N(i)}d_{ij} x_j$ (i.e., $q'_i$ is the same as $q_i$ except that we replace every $c_{ij}$ by $d_{ij}$).

\begin{claim}\label{cl:order}
For every $i$ and all $x,x'\in\{0,1\}^n$,
\[
    \operatorname{sgn}\bigl(q_i(x)-q_i(x')\bigr)
    =\operatorname{sgn}\bigl(q_i'(x)-q_i'(x')\bigr).
\]
\end{claim}

\begin{proof}[Proof of \cref{cl:order}]
We have
\[
  q_i(x)-q_i(x')=\sum_{j\in N(i)}c_{ij}(x_j-x'_j)
  \qquad\text{and}\qquad
  q_i'(x)-q_i'(x')=\sum_{j\in N(i)}d_{ij}(x_j-x'_j),
\]
where $x_j-x'_j\in\{-1,0,1\}$ and $|N(i)|\leq2$.  Since $\operatorname{sgn}(d_{ij})=\operatorname{sgn}(c_{ij})$, corresponding terms of the two sums have the same sign.  The sign of a sum of at most two terms is determined by the signs of the terms and, if these are opposite, by the comparison of their absolute values.  If two terms have opposite signs, then $i$ has two neighbors $j$ and $k$ with $x_j\neq x'_j$ and $x_k\neq x'_k$, so these absolute values are $|c_{ij}|$ and $|c_{ik}|$ in the first sum and $|d_{ij}|$ and $|d_{ik}|$ in the second.  Since the replacement preserves the comparison of absolute values, both sums have the same sign.
\end{proof}




\smallskip
\noindent\textbf{Compressing the unary coefficients.}
Fix $i$, and let $S_i:=\{q_i(x):x\in\{0,1\}^n\}$. Since $q_i$ depends only on the at most two variables $x_j$ with $j\in N(i)$, we have $|S_i|\le4$. Define $S_i'$ analogously
for $q_i'$. By Claim~\ref{cl:order}, the correspondence
$q_i(x)\mapsto q_i'(x)$ is well defined and preserves the order of all values;
denote it by $\varphi_i:S_i\to S_i'$.

We choose an integer $h_i'$ so that $h_i'+\varphi_i(s)$ has the same sign as
$h_i+s$ for every $s\in S_i$. There are four cases:
\begin{itemize}
\item if $-h_i\in S_i$, set
      $h_i':=-\varphi_i(-h_i)$;
\item if $-h_i<\min S_i$, set
      $h_i':=1-\min S_i'$;
\item if $-h_i>\max S_i$, set
      $h_i':=-1-\max S_i'$;
\item otherwise, let $a<b$ be consecutive elements of $S_i$ with
      $a<-h_i<b$, and set
      $h_i':=-\varphi_i(a)-1$.
\end{itemize}
Every element of $S_i'$ is even, since every $d_e$ is even. Thus, in the last
case, $\varphi_i(b)\ge\varphi_i(a)+2$, so $-\!h_i'$ lies strictly between
$\varphi_i(a)$ and $\varphi_i(b)$. It follows in all four cases that
\begin{equation}
\label{eq:222_equiv}
    \operatorname{sgn}\bigl(h_i+q_i(x)\bigr)
    =\operatorname{sgn}\bigl(h_i'+q_i'(x)\bigr)
\end{equation}
for every assignment $x$. Moreover,
$|h_i'|\le\max_{s'\in S_i'}|s'|+1
\le\sum_{e\ni i}|d_e|+1$.

\smallskip
\noindent\textbf{Bounding the length of an improving sequence.}
Define
\[
    \widetilde F(x)
      :=\sum_{i=1}^n h_i'x_i
        +\sum_{e=\{i,j\}\in E}d_e x_ix_j .    
\]

Equation~\eqref{eq:222_equiv} implies that for every assignment and every variable, flipping that variable increases $F$ if and only if it increases $\widetilde F$. 
Thus $F$ and $\widetilde F$ have exactly the same strictly improving flips. 
Since $\widetilde F$ is integer-valued, every improving flip increases its value by at least $1$, and
therefore every improving flip sequence has length at most $\max\widetilde F-\min\widetilde F$.

We have $\max\widetilde F-\min\widetilde F
 \le\sum_i|h_i'|+\sum_e|d_e|$.
The preceding bound on $|h_i'|$, summed over all $i$, yields $\sum_i|h_i'|\le2\sum_e|d_e|+n$, since every edge has two endpoints.
Consequently, $\max\widetilde F-\min\widetilde F\le3\sum_e|d_e|+n$.

Finally, we have $\sum_e |d_e| = 2\sum_e r_e \leq 2\left((m-s)s+\sum^s_{i=1}i \right) \leq 2\sum^m_{i=1} i \leq m(m+1)$.
Using $m\le n$, we obtain 
\[
    \max\widetilde F-\min\widetilde F
    \le3m(m+1)+n
    \le3n(n+1)+n
    =3n^2+4n .
\]
This proves the theorem.
\end{proof}

\section*{Statement on usage of generative AI}
We first developed a construction showing that $(2,2,6)$-MCA has the all-exp property. 
We then gave this construction to ChatGPT-5.6~Sol and asked whether the same underlying idea could be adapted to use at most five values per variable. 
This interaction led to the $(2,2,5)$-MCA construction presented in this paper. 
The proof originally suggested by ChatGPT-5.6~Sol was considerably more involved; the proof presented here was developed independently by the authors and is several pages shorter.

\bibliographystyle{plainurl}
\bibliography{refs}
\end{document}